\documentclass{article}
\usepackage{times,amsmath, epsfig, amssymb, mathtools}
\usepackage{graphics}
\usepackage{graphicx}
\usepackage{subfigure}
\usepackage{amsfonts, pifont}
\usepackage[ruled, linesnumbered, noend]{algorithm2e}
\usepackage{balance, cite}

\newtheorem{definition}{Definition}
\newtheorem{lemma}{Lemma}
\newtheorem{theorem}{Theorem}

\newtheorem{proof}{Proof}

\newenvironment{packed_items}{
\begin{itemize}
  \setlength{\itemsep}{1pt}
  \setlength{\parskip}{0pt}
  \setlength{\parsep}{0pt}
}{\end{itemize}}

\begin{document}

\title{Trajectory Based Optimal Segment Computation in Road Network Databases}

\author{
{Xiaohui Li{\small $~^{\#1}$}, Vaida \v{C}eikut\.{e}{\small $~^{*2}$}, Christian S. Jensen{\small $~^{*3}$}, Kian-Lee Tan{\small $~^{\#4}$}}%
\vspace{1.6mm}\\
\begin{minipage}[b]{0.45\linewidth}
\centering
\fontsize{10}{10}\selectfont\itshape
$~^{\#}$School of Computing\\
        National University of Singapore
\end{minipage}
\begin{minipage}[b]{0.45\linewidth}
\centering
\fontsize{10}{10}\selectfont\rmfamily\itshape
$~^{*}$Department of Computer Science\\
        Aarhus University
\end{minipage}
\vspace{1.6mm}\\
\begin{minipage}[b]{0.45\linewidth}
\centering
\fontsize{9}{9}\selectfont\ttfamily\upshape
$~^{1, 4}$\{lixiaohui, tankl\}@comp.nus.edu.sg%
\end{minipage}
\begin{minipage}[b]{0.45\linewidth}
\centering
\fontsize{9}{9}\selectfont\ttfamily\upshape
$~^{2, 3}$\{ceikute, csj\}@cs.au.dk
\end{minipage}
}
\date{}
\maketitle

\begin{abstract}
  Finding a location for a new facility such that the facility
  attracts the maximal number of customers is a challenging
  problem. Existing studies either model customers as static sites and
  thus do not consider customer movement, or they focus on theoretical
  aspects and do not provide solutions that are shown empirically to
  be scalable.
  Given a road network, a set of existing facilities, and a collection
  of customer route traversals, an optimal segment query returns the
  optimal road network segment(s) for a new facility. We propose a
  practical framework for computing this query, where each route
  traversal is assigned a score that is distributed among the road
  segments covered by the route according to a score distribution
  model. The query returns the road segment(s) with the highest score.
  To achieve low latency, it is essential to prune the very large
  search space. We propose two algorithms that adopt different
  approaches to computing the query. Algorithm AUG uses graph
  augmentation, and ITE uses iterative road-network partitioning.
  Empirical studies with real data sets demonstrate that the
  algorithms are capable of offering high performance in realistic
  settings.
\end{abstract}

\section{Introduction}

The problem of finding a location for a new facility with respect to
given sets of customer locations and existing facilities, known as the
\emph{facility location}
problem~\cite{CDL05, DZX05, WOY09,  XYL10, ZDX06, ZWL10, FH09, NP05, CDL10}, has
applications in the strategic planning of resources (e.g., hospitals,
gas stations, banks, ATMs, billboards, and retail facilities) in both
the public and private sectors~~\cite{JOD07a, JOD07b}.
The literature contains a line of study that use the residences of
consumers as the customer locations~\cite{WOY09, XYL10, ZWL10}.
However, customers do not remain stationary at their residences, but
rather travel, e.g., to work. Consumers are not only attracted to
facilities according to the proximity of these to their residences.

Another line of study~\cite{AB96,B95,B97,BKX95,ABK09,BBL95} considers
the \emph{flow intercepting facility location} problem, where the goal
is to identify a location that intercepts the most flow from moving
customers.
Flows are made up by pre-planned customer trips, and the idea is that
customers can choose to interrupt their trip to receive a service from
a facility at a nearby location.
In its original formulation, the problem is to maximize the flow in a
network while placing $m$ new facilities while disregarding existing
facilities. Studies of this problem have a theoretical focus and do
not focus on providing scalable solutions. Thus, the largest study
considers spatial networks with up to 1,000 nodes~\cite{ABK09}. Real
spatial networks for even small regions are much larger. Another
difficulty is to obtain real flow data. This led to the development of
probabilistic methods~\cite{BKX95}.

The increasing availability of moving-object trajectory data, e.g., as
GPS traces, calls for a new study of the facility location problem
that takes into account the real movements of the customers that are
now available and that provides practical solutions that apply in
realistic settings.

We study the \emph{optimal segment} problem.  Given a road network
$G$, a set of facilities $F$, a set of route traversals $R$, each of which can be taken by
different users multiple times, the objective is to find the optimal
road segments such that a new facility on any of these segments
attracts the maximum number of route traversals. A route traversal is
\emph{attracted} by a facility if the distance between the route and
the facility is within a given threshold.

Figure~\ref{fig:ex_1} shows an instance of the problem.  Solid lines
and dots form the road network.  Hollow circles are existing
facilities ($f_1, f_2, f_3$, and $f_4$). Dashed lines indicate route
traversals ($r_1, r_2$, and $r_3$). We draw them next to the roads
for clarity.
The gray bar that covers $f_3$ indicates that $r_3$ is attracted by
$f_3$ because $f_3$ is within distance $\delta$ of one of the end
points of $r_3$.  The rationale is that a facility that is
sufficiently near a route will attract customers who follow the route.
Therefore, the ends of each route are extended by distance $\delta$.

\begin{figure}[!h]
  \centering
    \includegraphics[width=0.68\textwidth]{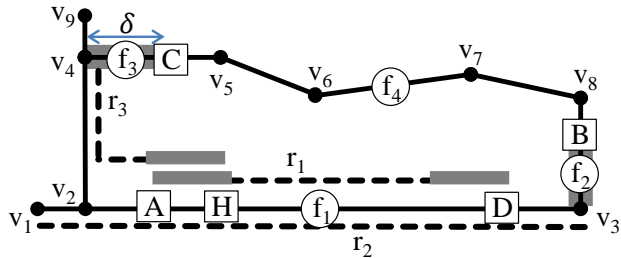}
\vspace{-0.5cm}
  \caption{Optimal Segment Problem Example}
  \label{fig:ex_1}
\end{figure}

With $\delta$, $r_1$ starts and ends at $A$ and $D$, respectively.
Route $r_2$ starts and ends at $v_1$ and $B$, respectively.  Route
$r_3$ starts and ends at $C$ and $H$, respectively.  Assume that each
of the routes is traversed by one customer exactly once. Intuitively,
the optimal segment for a new facility is the segment
$\mathit{\overline{AH}}$ because this segment attracts the most route
traversals (in this example, three).

We propose a framework to solve the optimal segment problem.  In the
framework, each route traversal is assigned a score, and that score is
distributed among the road network segments covered by the
traversal. The scoring of segments is based on three factors: the
number of customers who take the route (the count), the number of
traversals by each customer (the usage), and the length of the route.

Intuitively, road segments that are covered by many route traversals
and that are attracted by few existing facilities are good result
candidates.  But customers of different types of businesses can have
different spatial preferences with respect to the businesses they are
likely to visit. For example, customers may prefer grocery stores near
their homes or work places, but may have equal probability to visit
clothing stores along the routes they travel. To accommodate such
preferences, we support different functions for the assignment of
scores to the routes that customers follow as well as allow different
models for the distribution of scores to the underlying segments.

The framework encompasses two optimal segment algorithms. The first,
AUG, uses graph augmentation, the idea being to augment the set of
vertices of the original road network graph with the facilities and
the start and end points of the route traversals.  Each vertex in the
new graph records a list of attracted routes.  The score of an edge is
the sum of scores of the route traversals that cover both vertices of
an edge.  The edges with the highest score are mapped back to the
original graph and are possibly extended into longer segments.

The second, ITE, uses a heap to prioritize the most promising road
segments, and it iteratively partitions and scores these based on
intersecting routes. ITE keeps partitioning the road segments that
most likely contain an optimal subsegment until an optimal subsegment
is obtained.  Then it extends the partial optimal segment to its full
length and adds it to the result set.

In summary, the contribution is fourfold: 
\begin{packed_items}
\item Formalization of the new \emph{optimal segment} problem.
\item A framework that accommodates different scoring functions and score distribution models.
\item Two algorithms, AUG and ITE, that solve the problem.
\item Coverage of an empirical study that indicates that AUG and ITE are efficient in realistic settings.
\end{packed_items}

The remainder of the paper is structured as follows.
Section~\ref{sec:def} formalizes the problem setting.
Section~\ref{sec:preprocess} presents a preprocessing procedure that
is used by both of the two proposed algorithms.  We describe in detail
algorithm AUG and provide a theoretical analysis in
Section~\ref{sec:aug}.  We then describe in detail algorithm ITE and
give an accompanying theoretical analysis in Section~\ref{sec:ite}.
Section~\ref{sec:expt} reports the results of an empirical evaluation
of the proposed algorithms.  Section~\ref{sec:related} reviews
existing work. Finally, Section~\ref{sec:conclusion} concludes.

\section{Definitions}
\label{sec:def}

We proceed to model the road network and formulate the optimal segment problem along with supporting definitions.

\subsection{Road Network Modeling}

A road network is modeled as a \emph{spatially embedded graph} $G=(V, E)$,
where $V$ is a set of vertices, and $E$ is a set of edges that connect ordered pairs of vertices.
Every vertex $v_i$ has $(x_i, y_i)$ coordinates in 2D space, denoted as $\mathit{loc}(v_i)=(x_i , y_i)$.
We use either $e_{i,j}$ or $(v_i, v_j)$ to refer to the directed edge from vertices $v_i$ to $v_j$.
The length of an edge is defined as the Euclidean distance between its two vertices: $\|e_{i,j}\|=\|\mathit{loc}(v_i),\mathit{loc}(v_j)\|$.
Vertices and edges are assigned unique identifiers.
In this model, an edge between two vertices represents a part of a road.
The polyline obtained by connecting the vertices of consecutive edges approximates the center line of part of a road.

We use the term \emph{network point} to refer to a point location anywhere on an edge.

\begin{definition}(Network Point)
A road network point $p$ is defined as $p=(\mathit{eid}, d)$,
where $\mathit{eid}$ is the identifier of an edge $e=(v_i,v_j)$ and $d$ ($0 \leq d \leq 1$)
is the ratio of the distance between vertex $v_i$ and the point to the length of $e$.
$P$ denotes the set of all network points on the road network.
\end{definition}

It can be seen that given an edge $(v_i, v_j)$ identified by $\mathit{eid}$, $v_i=(\mathit{eid}, 0)$ and $v_j=(\mathit{eid}, 1)$.
Therefore we have $V \subset P$.
For example, in Figure~\ref{fig:ex_1}, the network point of $f_1$ is ${f_1}.p=(e_{2,3}, 0.5)$.
The distance between two network points $p_i$ and $p_j$ on the same edge $e$ is defined as $\mathit{dist}(p_i, p_j) = \|e\|\cdot |d_i - d_j|$.

A road segment is a polyline that starts at a network point, traverses a sequence of vertices, and ends at a network point.

\begin{definition}(Road Segment)
A road segment $s$ is defined as a sequence of network points,
$s = \langle p_1, p_2, \ldots, p_n \rangle$,
where $n \geq 2, \, p_1, p_n \in P, \, p_i \in V, p_1.\mathit{eid} = p_2.\mathit{eid},$ $p_{n-1}.\mathit{eid} = p_n.\mathit{eid}$
and $(p_i, \, p_{i+1}) \in E \, (1 < i < n-1)$.

The length of $s$ is the \emph{network distance} from $p_1$ to $p_n$.
\[
\mathrm{length}(s) =
\left\{
  \begin{array}{l l}
    \mathit{dist}(p_1, p_2) &  \text{$n = 2$}\\
    \mathit{dist}(p_1, p_2)+ \displaystyle \sum_{i=2}^{n-2} \| e_{i, {i+1}}\|  + \mathit{dist}(p_{n-1}, p_n) &  \text{$n > 2$}\\
  \end{array} \right.
\]

The set of road segments is denoted as $S$.
\end{definition}

It follows from definition that an edge is also a segment, i.e., $E \subset S$.
Further, we use the notion \emph{route} for a segment that a customer has traversed.

When there is no ambiguity from the context, we use $\mathit{\overline{AB}}$ to mean the segment between network points $A$ and $B$.
For example, the short segment between $A$ and $H$ in Figure~\ref{fig:ex_1} is $\mathit{\overline{AH}}$.
Otherwise, we write the segment in full, e.g., the road segment $\langle p_{f_3}, v_5, v_6, p_{f_4} \rangle$
between facilities $f_3$ and $f_4$,
supposing the network points for facilities $f_3$ and $f_4$ are $p_{f_3}$ and $p_{f_4}$, respectively.

\subsection{Facilities and Route Usage}
\label{subsec:routesAndFac}

A facility $f$ located at a network point $p$ is denoted as $(\mathit{fid}, p)$, where $\mathit{fid}$ identifies the facility.
$F$ denotes the set of all facilities.

A route is a segment and thus starts at a network point, traverses a sequence of connected edges, and stops at a network point.
The same route can be traversed many times by the same or many customers.
For instance, many customers who live in the same building may take the same route $r$ to the same grocery store.
We use $\mathit{count}$ to denote the number of customers who take $r$.
On the other hand, one customer can take the same route many times,
e.g., a customer may take the same route from home to work on most weekdays.
We use $\mathit{usage}_i$ to denote the number of times $r$ is taken by customer $i$
($1 \leq i \leq \mathit{count}$).

\begin{definition} (Route Usage Object)
A route usage object $\mathit{ro}$ is defined as
\[
\mathit{ro} = (\mathit{rid}, r,
    \mathit{count}, \langle \mathit{usage}_1, \ldots, \mathit{usage}_{\mathit{count}} \rangle)
\]
where $\mathit{rid}$ identifies the object, $r\in S$ is a segment traversed by the user. $R$ is a set of all route usage objects.
\end{definition}

A route $\mathit{ro}.r$ \emph{covers} a road segment $\mathit{s'}$ if $\forall p  \in \mathit{s'}(p \in \mathit{ro}.r)$.
A route $\mathit{ro}.r$ \emph{intersects} a segment $\mathit{s'}$ if
$\exists p \in \mathit{s'}(p \in \mathit{ro}.r)$.
The set of route usage objects whose routes cover $\mathit{s'}$ is denoted as $\mathit{s'}.C$.
The set of route usage objects whose routes intersect $\mathit{s'}$ is denoted as $\mathit{s'}.I$.
It is straightforward to see that $\mathit{s'}.C \subseteq \mathit{s'}.I$.
We also say that $\mathit{ro}_1 \equiv \mathit{ro}_2$ if $\mathit{ro}_1.r = \mathit{ro}_2.r$.

In Figure~\ref{fig:ex_1}, the routes $r_1, r_2$, and $r_3$ are
traversed by three different customers.
We assume that each route is traversed once by each customer.

A route $r$ is \emph{attracted} by a facility $f$ and $f$ is an \emph{attractor} for $r$
if $\mathit{dist}_G(f.p, \mathit{ro}.r) \leq \delta$, where
$\mathit{dist}_G(p, s)$ gives the shortest {network distance} between a network point $p$ and a segment $s$
and $\delta$ is the distance threshold that was introduced earlier.
Note that the same facility can attract several routes.
In Figure~\ref{fig:ex_1}, facility $f_1$ attracts routes $r_1$ and $r_2$, and $f_3$ attracts $r_3$.

\subsection{Scoring a Route}
\label{subsec:scoring}

In the optimal segment problem, route traversals play the role that customer locations play in the classical formulation of the optimal location problem.
Thus, we need to decide how to assign a score to a route based on the traversals of the route.
The scoring of a route is thus based on three factors that are all captured in the route usage object for the route:
the number of customers taking the route,
the number of traversals by each customer, and the length of the route.
The route's score is subsequently distributed among the segments covered by the route.
The intuition of distributing the score of a route to its segments is that when a customer traverses the route,
the customer may visit facilities located on segments along the route.

To ensure that the framework yields meaningful results, the scores eventually assigned to segments must be
invariant under the splitting and concatenation of route usage objects.
To achieve this, we require the following property to hold.

\textbf{Route Scoring Property} \quad
A route scoring function should be independent of the partitioning of the route of a route usage object.
Let $\mathit{ro}_1 \circ \mathit{ro}_2$ be the concatenation of $\mathit{ro}_1.r$ and $\mathit{ro}_2.r$.
Let $\mathit{ro}=(\mathit{id}, r_1 \circ r_2 \circ \cdots \circ r_m, \mathit{count}, u)$ and
$\mathit{ro}_i=(\mathit{id}_i, r_i, \mathit{count}, u) \, (1 \leq i \leq m)$.
Then we require
$
\mathit{score}(r) = \sum_{i=1}^{m}{\mathit{score}(r_i)}
$.

This property ensures that partitioning a route usage object does not change the total score that is available for assignment to segments.

Many scoring functions are possible that satisfy the property.
Next, we show two of them.

\begin{definition}(Scoring a Route)
Let a route $r$ with an associated route usage object $\mathit{ro}=(\mathit{rid}, r, \mathit{count}, \langle \mathit{usage}_1, \dots, \mathit{usage}_{\mathit{count}} \rangle)$ be given. Then the score of $r$ can be defined as follows
\begin{align*}
\mathit{score}_{\mathit{all}}(r) &= \mathrm{length}(r) \sum_{i=1}^{\mathit{count}} \, \mathit{ro.usage_i} \\
\mathit{score}_{\mathit{cap}-x}(r) &= \mathrm{length}(r) \sum_{i=1}^{\mathit{count}}{\min(\mathit{ro}.\mathit{usage}_i, x)}
\end{align*}
where $x$ is a user-defined value.
\end{definition}

Depending on the products or services offered by a facility, different scoring functions may be appropriate. For example, a facility that sells everyday necessities (e.g., a bakery) may attract the same customer on each route traversal by the customer. Thus, $\mathit{score}_{\mathit{all}}$ is appropriate. In contrast, if a store sells products that are bought less frequently (e.g., a furniture store), the store may not benefit from a large number of traversals by the same customer, making $\mathit{score}_{\mathit{cap}-x}$ more appropriate. Thus, we keep the framework open to the use of different scoring functions.

Unless specified otherwise, we use the function $\mathit{score}_{\mathit{all}}$ for illustration.

In Figure~\ref{fig:ex_1}, assuming that the lengths of routes $r_1, r_2$, and $r_3$
are $2$, $4$, and $3$, and the number of traversals per customer are $\langle 2, 2 \rangle$, $\langle 2, 1 \rangle$, and $\langle 2 \rangle$, respectively.
Then we have three route usage objects:
$\mathit{ro}_1=(\mathit{id}_1, r_1, 2, \langle 2,2 \rangle)$,
$\mathit{ro}_2=(\mathit{id}_2, r_2, 2, \langle 2,1 \rangle)$, and
$\mathit{ro}_3=(\mathit{id}_3, r_3, 1, \langle 2 \rangle)$.
The score of $r_1$ can be calculated as follows,
\(
\mathit{score}(r_1) = \mathrm{length}(r_1) \cdot (\mathit{ro}_1.\mathit{usage}_1+\mathit{ro}_2.\mathit{usage}_2) = 2 \cdot (2+2) = 8
\)
Similarly, we calculate the scores of $r_2$ and $r_3$, $\mathit{score}(r_2)= 12$ and $\mathit{score}(r_3) = 6$.

\subsection{Score Distribution Models}
\label{subsec:interest}

A score distribution model determines how to distribute the score of a route to the underlying segments.

Intuitively, segments covered by a route with many traversals that
are \emph{not attracted} by many other facilities are good candidates for placing a new facility.
Therefore, they should be assigned high scores.
But customers can have different spatial preferences for visiting different kinds of businesses.
Therefore, we leave the framework open to the use of different score distribution models.

When $n$ facilities are located on a segment, they partition the segment into $k$ subsegments
where $k$ is one of $n-1, n$, or $n+1$ depending on whether two, one, or no facilities are located at the ends of the segment.

The following are example score distribution models.

\begin{itemize}
\item Equal weight is assigned to each subsegment. In this model, the
  score of a route $r$ is distributed such that a customer has an
  equal probability to visit any business along the route. For
  example, any clothing store on the way back home.  The score
  assigned to the $i$th subsegment $s_i$ ($1 \leq i \leq k$) is
  $\frac{1}{k} \cdot \mathit{score(r)}$.

\item Decreasing/increasing weights are assigned to the subsegments.
  The score of the $i$th subsegment $s_i$ ($1 \leq i \leq k$) is given
  by $\frac{1}{2^i} \cdot \frac{1}{\sum_{i=1}^{k}{\frac{1}{2^i}}}
  \cdot \mathit{score}(r)$, ($1 \leq i \leq k$).  This definition
  gives exponentially decreasing scores to subsegments and normalizes
  the scores such that the full score of the route is distributed.
  This model indicates a preference for the facilities at the
  beginning of the route.  Symmetrically, there is a model that
  prefers the facilities at the end of the route.  For example, a
  customers might prefer to have a meal before the trip back home or
  to work, but it is also possible (with lower probability) that the
  customer will visit any restaurant along the route

\item All of the score is evenly distributed to the first and the last
  subsegment.  This model indicates that customers consider only
  businesses that are located nearest to the route destinations.  For
  example, a customer would like to visit the store, which sells dairy
  products, closest to home, but for regular items closest store to
  work place can be used.

\item The original \emph{facility location} problem considers simply
  the attraction of customer locations to facility locations.  In our
  setting, where customer route traversals are attracted to segments
  where facilities may be placed, the model that assigns the entire
  score of a route traversal to the route's first subsegment may be
  the one that most closely resembles the original problem.
\end{itemize}

In Figure~\ref{fig:ex_1}, route $r_2$ is attracted by facilities $f_1$ and $f_2$, and $k=3$.
In previous examples, we showed that the score of $r_2$ is 12.
According to the first proposed score distribution model, each subsegment ($\overline{v_1f_1}, \overline{f_1f_2}$, and $\overline{f_2B}$)
receives score $\frac{\mathit{score}(\mathit{r}_2)}{k}=\frac{12}{3} = 4$.
According to the second model,
$\frac{1}{\sum_{i=1}^{\mathit{k}} \, \frac{1}{2^i}} = \frac{1}{\frac{1}{2} + \frac{1}{4} + \frac{1}{8}}=\frac{1}{\frac{7}{8}}=\frac{8}{7}$,
$\mathit{score}(\overline{v_1f_1}) = \frac{1}{2} \cdot \frac{8}{7} \cdot 12 = \frac{4}{7} \cdot 12$,
$\mathit{score}(\overline{f_1f_2}) = \frac{1}{4} \cdot \frac{8}{7} \cdot 12 = \frac{2}{7} \cdot 12$,
and $\mathit{score}(\overline{f_2B}) = \frac{1}{8} \cdot \frac{8}{7} \cdot 12  = \frac{1}{7} \cdot 12$.

So far, we have distributed the score assigned to a single route to the segments covered by the route.
However, a segment $s$ may be covered by multiple routes that assign score to the segment.
The total score of the segment, $\mathit{score}_M(s)$, where $M$ indicates the score distribution model used, is simply the sum of these scores.
Similarly, we can calculate the score of a network location $p$, $\mathit{score}_M(p)$.

We show how to score the subsegment $\overline{AH}$ in Figure~\ref{fig:ex_1} using the first proposed model.
$\overline{AH}$ is covered by all of the three route usage objects.
So $\mathit{score}_M(\overline{AH})= \sum_{i=1}^3{\frac{\mathit{score}(\mathit{ro}_i.r)}{k_i}} = \frac{8}{2}+\frac{12}{3}+\frac{6}{2} = 11$.

Since our framework is generic w.r.t. score distribution models, unless specified otherwise, we use the first model for illustration.

\subsection{Problem Formulation}
With the above definitions in place, we can define the optimal segment query.

\begin{definition}(The Optimal Segment Query)
The optimal segment query finds every segment $s_{\mathit{opt}}$ from a road network $G$
such that
\begin{enumerate}
\item $\forall p_1, p_2 \in s_\mathit{opt} \, ( \mathit{score}_M(p_1) = \mathit{score}_M(p_2) )$
\item  $\forall p \in s_\mathit{opt} \, \forall p' \in P \, (\mathit{score}_M(p') \leq \mathit{score}_M(p))$
\item $\nexists s' \in S \, (s_\mathit{opt} \subset s'$ and 1 and 2 hold$)$.
\end{enumerate}
\end{definition}

This definition ensures that every point in the optimal segment has the same score, that the score is optimal, and that the optimal segment is maximal.

Notation used introduced this section and to be used throughout
the paper is summarized in Table~\ref{tbl:symbols}.

\begin{table}[th]
\centering
    \begin{tabular}{ | l | p{6cm} |}
    \hline
    $R$ & The set of route usage objects\\ \hline
    $F$ & The set of facilities\\ \hline
    $S$ & The set of road segments \\ \hline
    $P$ & The set of sites \\ \hline
    $G$, $G'$ & The (augmented) road network graph\\ \hline
    $V$, $V'$ & The set of vertices in $G$, $G'$\\ \hline
    $E$, $E'$ & The set of edges in $G$, $G'$\\ \hline
    $n$ & The total number of GPS points in $R$ \\ \hline
    $\delta$ & The maximum distance of attraction  \\ \hline
    \end{tabular}
    \caption{Summary of Notation}
  \vspace{-0.2cm}
  \label{tbl:symbols}
\end{table}

\section{Preprocessing}
\label{sec:preprocess}
A straightforward approach to compute the optimal segment query is to enumerate and score all possible segments and then return the one with the highest score.
However, this is not feasible as there is an infinite number of possible segments.
Thus, different approaches are needed.

The two algorithms we propose both rely on the same preprocessing algorithm, which we present here.
This algorithm determines the relationships between the facilities and the edges,
between the routes and the edges, and between the facilities and the routes.
It needs to be run only once for one set of routes.

The algorithm makes each edge record its facilities and
route start and end points, if any.
It also makes each vertex record the covering routes' identifiers.
The routes record the facilities they cover.
The facilities record the edge they are located on and the covering routes, if any.
Also the algorithm populates a lookup table so that given an edge, one can quickly determine
the routes that intersect with the edge.

Recall that $G$ is the spatially embedded graph, $f$ is a facility, and $r$ is a route.
Algorithm PreProcess calls getEdge($f, G$) to retrieve the edge where $f$ is located.
It also calls getEdges($r, G, \delta$) to retrieve the set of edges that intersect $r$.

The PreProcess procedure is presented in Algorithm~\ref{alg:PreProcess} and explained next.

A facility $f$ keeps the edge where it is located in the variable $f.e_c$, and the set of routes it attracts in $f.R_c$.
An edge $e$ keeps a set of route start and end network points that are located on $e$ in $e.O_c$.
This list is used by the AUG algorithm for augmentation purpose.
Each vertex $v$ of $e$ maintains a list of routes that it attracts in $v.R_c$,
and $v$'s relative positions in $v.L$.
These two lists are used later by AUG for scoring purpose.
$e.F_c$ and $e.R_c$ are the set of facilities that are located on $e$ and the set of routes that intersect edge $e$, respectively.
A route $r$ keeps its set of attracting facilities in $r.F_c$.

For each facility $f$, PreProcess retrieves its edge $e$
so that $e$ adds $f$ to its set of facilities $e.F_c$, and $f.e_c$ is set to the right edge (line~1).

\begin{algorithm}[ht]
\SetKwData{Left}{left} \SetKwData{This}{this} \SetKwData{Up}{up}
\SetKwFunction{Union}{Union}
\SetKwFunction{FindCompress}{FindCompress} \SetKwInOut{Input}{input}
\SetKwInOut{Output}{output}
\SetKwIF{If}{ElseIf}{Else}{if}{then}{else if}{else}{endif}
\SetKwFor{For}{for}{do}{endfor} \SetKwFor{While}{while}{do}{endw}
\SetKwFor{ForEach}{foreach}{do}{endfch}
\caption{PreProcess($G, R, F, \delta$) }
\label{alg:PreProcess}

\ForEach{$f \in F$}{
    $e \leftarrow \mathrm{getEdge}(f, G)$; 
    $e.F_c.\mathrm{add}(f)$; 
    $f.e_c \leftarrow e$;
}
\ForEach{$\mathit{ro} \in R $}{
    $r \leftarrow $ $\mathit{ro}.r$; \\
    $r.E_c \leftarrow \mathrm{getEdges}(r, G, \delta)$; \\
   \ForEach{$e=(v_s, v_e) \in r.E_c$} {
			\textbf{if} $(r.p_s.\mathit{eid} = e.\mathit{eid}) \wedge (r.p_s.d\not \in \{0, 1\})$ \textbf{then} $e.O_c.\mathrm{add}(r.p_s)$; \\
      	\textbf{if} $(r.p_e.\mathit{eid} = e.\mathit{eid}) \wedge (r.p_e.d\not \in \{0, 1\})$ \textbf{then} $e.O_c.\mathrm{add}(r.p_e)$; \\
        $e.R_c.\mathrm{add}(r)$; \\
        \If{$\mathrm{contains}(r, e)$}{
            $r.F_c \leftarrow r.F_c \cup e.F_c$; \\
            \textbf{foreach} $f \in e.F_c$ \textbf{do}  $f.R_c.\mathrm{add}(r)$; \\
        }
        \ElseIf{$\mathrm{intersects}(r, e)$}{
            $F' \leftarrow \{f| f\in e.F_c \wedge \mathrm{attracts}(f, r, \delta)\} $; \\
            $r.F_c \leftarrow r.F_c \cup F'$; \\
            \textbf{foreach} $f \in F'$ \textbf{do}  $f.R_c.\mathrm{add}(r)$; \\
        }
        \If{$\mathrm{attracts}(v_s, r, \delta)$}
        {
            $v_s.R_c.\mathrm{add}(r)$; 
            $i \leftarrow$ the position of $v_s$ relative to the $r.F_c$; 
            $v_s.L.\mathrm{add}(i)$;
        }
        \If{$\mathrm{attracts}(v_e, r, \delta)$}
        {
            $v_e.R_c.\mathrm{add}(r)$; 
             $i \leftarrow$ the position of $v_e$ relative to the $r.F_c$; 
            $v_e.L.\mathrm{add}(i)$;
        }
       }
}
\end{algorithm}

Next, for each route $r$, the set of intersected edges is retrieved,
and the $r.E_c$ field is updated (line~5).
Then, if the start network point of $r$ is not a vertex in $G$, it is added
to the $e.O_c$ set of the edge $e$ where it is located.
Similarly, $r$'s end network point is added to a $e.O_c$ set. (lines~7--8).
Route $r$ is also added to the list $e.R_c$ (line~9).
For each edge $e$ covered by the route $r$, the facilities and $r$ record each other (lines~10--12).
For each edge $e$ intersected by a route $r$, on the other hand, the attraction relationship between
the facilities and $r$ is determined before updating each other's corresponding field (lines~13--16).
Next, each vertex of $e$ records $r$ in the $v.R_c$ list if $r$ is attracted by it.
In addition, the relative position of $v$ is also kept in $v.L$ for scoring purpose (lines~17--20).

In Figure~\ref{fig:ex_1}, edge $e_{2,3}$ has $e_{2,3}.F_c=\{f_1\}$ and $f_1.e=e_{2,3}$.
Route $r_1$ traverses one edge and is attracted by one facility, so, $r_1.E_c=\{e_{2,3}\}$ and $r_1.F_c=\{f_1\}$. Edge $e_{2,3}$ is covered by $r_1, r_2$ and $r_3$, so, $e_{2,3}.R_c=\{r_1,r_2,r_3\}$.
Three start or end network points of the routes are located on edge $e_{2,3}$, so, $e_{2,3}.O_c=\{A,D,H\}$.
Vertex $v_2$ is covered by $r_2$ and $r_3$ and vertex $v_3$ is covered $r_2$, so $v_2.R_c=\{r_2, r_3\}$ and $v_3.R_c=\{r_2\}$.
Facility $f_1$ attracts $r_1$ and $r_2$, so $f_1.R_c=\{r_1, r_2\}$.

The following lemma states the time complexity of PreProcess.

\begin{lemma}
Algorithm PreProcess has time complexity $O(|F| + |R||E_m|)$,
where $|E_m|$ is the maximal number of edges that any route traverses.
\end{lemma}

\begin{proof}
The first loop in the algorithm takes time $O(|F|)$.
In the second loop, the outer loop runs $O(|R|)$ times.
The inner loop depends on the number of edges that a route traverses.
Let $|E_m|$ be the maximal number of edges that any route traverses.
Then the second loop has time complexity $O(|R||E_m|)$.
In total, the time complexity is $O(|F| + |R||E_m|)$.
\end{proof}

\section{Graph Augmentation}
\label{sec:aug}

\subsection{Overview}
The main idea of the graph augmentation algorithm (AUG) is to
augment the road network graph $G$ with the facilities and
the first and the last network points of each route.
In the augmented graph $G' = (V', E')$ it is guaranteed that each route starts from
a vertex and ends at a vertex.
Meanwhile, each vertex in $G'$ stores the identifiers of the covering routes.

Then each edge's score in $G'$ can be calculated
by summing up the scores distributed by the routes that cover both ends points.
The score contributed by a route is calculated based on the specific score distribution model used, as discussed in Section~\ref{subsec:interest}.

Next, AUG examines every edge in $G'$ with a score,
and identifies the edges with the highest score (the optimal edges).

Finally, the algorithm maps the optimal edges back to the original graph $G$, where they are segments.
Then AUG merges connected segments, if any, to form maximal segments,
and returns them as the result.

Figure~\ref{fig:augment} illustrates the graph in Figure~\ref{fig:ex_1}
after being augmented with routes $r_1$, $r_2$, and $r_3$
and facilities $f_1, \, f_2, \,$ and $f_3$.
Note that each vertex in the augmented graph has a list of the  identifiers of the routes that cover the vertex.
We use $\mathit{AL}(v_i)$ to denote the attraction list of $v_i$.
Intersecting the sets of two adjacent vertices gives
the routes that cover the edge, whose score can then be calculated
according to a score distribution model.
For example, $\mathit{AL}(A) = \{r_1, r_2, r_3\}$ and $\mathit{AL}(H) = \{r_1, r_2, r_3\}$.
So, the set of routes that cover edge $e_{A,H}$ is
$\mathit{AL}(A) \cap \mathit{AL}(H) = \{r_1, r_2, r_3\}$.
Then the score of $e_{A,H}$ is calculated based on the score distribution model used.

\begin{figure}[ht]
  \centering
    \includegraphics[width=0.68\textwidth]{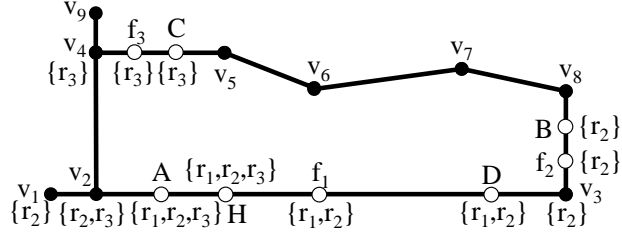}
\vspace{-0.3cm}
  \caption{The Augmented Road Network Graph}
  \label{fig:augment}
\end{figure}

Next, AUG finds the edges with the highest score by
examining all edges in the augmented graph.
These edges are then mapped back to the original graph,
and become road segments, which are possibly merged into longer segments.
These segments are returned as the result.
In Figure~\ref{fig:augment}, after edge $e_{A,H}$ is identified
as the optimal edge with the highest score, it is mapped back to the original
graph, and the segment $\overline{\mathit{AH}}$ is returned as the result.

\subsection{The AUG Algorithm}

Algorithm~\ref{alg:findSegAug} presents details of the AUG algorithm.
The set of edges that have network points either from facilities or routes is obtained (lines~1--2).
Graph $G'$ is obtained by augmenting graph $G$ with the network points of $F$ and $R$ (line~3).
Note that some network points of routes or facilities may happen to be vertices.
These network points are excluded from being augmented into $G$.
Then AUG updates the covering routes of the newly added vertices (lines~4--6).
It initializes the result set $S$ and the highest score seen so far, $\mathit{optS}$ (line~7).
In the next loop (lines~8--14), the scores of the edges in $G'$ are calculated according to the score distribution model used (line~10),
and the optimal edges in $G'$ are identified and stored in $S$.
In lines~15--17, each edge in $G'$ is mapped back to $G$.
Mapping back to the original graph is a trivial task.
Recall that each new network point has an $\mathit{eid}$ field that helps
identify the original edge.
If a segment can be extended (i.e., the neighboring segment is also an optimal segment),
it is extended (line~17).
Finally, the result set $S$ is returned (line~18).

\begin{algorithm}
\SetKwData{Left}{left} \SetKwData{This}{this} \SetKwData{Up}{up}
\SetKwFunction{Union}{Union}
\SetKwFunction{FindCompress}{FindCompress} \SetKwInOut{Input}{input}
\SetKwInOut{Output}{output}
\SetKwIF{If}{ElseIf}{Else}{if}{then}{else if}{else}{endif}
\SetKwFor{For}{for}{do}{endfor} \SetKwFor{While}{while}{do}{endw}
\SetKwFor{ForEach}{foreach}{do}{endfch}
\caption{AUG($G, R, F, \delta, M$) }
\label{alg:findSegAug}
$E_F \leftarrow $ the set of edges where $F$ are located; \\
$E_R \leftarrow $ the set of edge that $R$ intersect; \\
$G' \leftarrow \mathrm{Augment}(G, F, R)$; \\

\ForEach{$e \in (E_F\cup E_R)$}{
    \ForEach{$v_i \in e.F_c \cup e.O_c$}{
        $v_i.R_c \leftarrow \mathrm{getCoverRouteIds}(v_i, e.R_c)$; \\
    }
}

$S \leftarrow \emptyset$;
$\mathit{optS} \leftarrow 0 $; \\

\ForEach{$e=(v_s, v_e) \in G'.E'$}{
    $R' \leftarrow v_s.R_c \cap v_e.R_c;$ \\
    $\mathit{score}_M(e) \leftarrow$ compute the score of $e$ based on $M$; \\
    \If{$\mathit{score}(e) > \mathit{optS}$}{
        $\mathit{optS} \leftarrow \mathit{score}(e)$; \\
        $S \leftarrow \{ e \}$;
    }
    \textbf{else if} $\mathit{score}(e) = \mathit{optS}$ \textbf{then}  $S \leftarrow S.\mathrm{add}(e)$;
}

\ForEach{$s \in S$}{
    map $s$ to $G$; \\
    \textbf{if} $\mathrm{canExtend}(s, G)$ \textbf{then}  $\mathrm{extend}(s, G)$;
}
return $S$;
\end{algorithm}

This process has two implications.
First, an edge in $G$ may be split into several edges in $G'$.
After the optimal edges are identified in $G'$, they must be
mapped back to $G$.
Second, for an edge in $G'$, a route either covers it or does not cover it.
The \emph{partial intersection} relationship
between a route and an edge is eliminated in $G'$.

It can be seen that it is sufficient to
just augment the original graph with
the start point and the end point of each route for finding
the optimal segments
because the internal points in a route are vertices in $G$.

The algorithm splits some road segments.
\begin{packed_items}
\item
If the two end points of a route do not happen to be vertices in $G$, they are added as new vertices into the road network,
as they are covered by at least one route.
\item
If a facility does not happen to be a vertex in $G$, it is added as a new vertex if it attracts any route,
e.g., $f_1$, $f_2$, and $f_3$ in Figure~\ref{fig:augment}.
Facility $f_4$ no longer exists in the augmented graph because it does not attract any routes.
\item
In order to accommodate these new vertices, some edges in $G$ are replaced with ``smaller" edges in $\mathit{G'}$.
    For example, in Figure~\ref{fig:augment}, the edge $e_{2,3}$ is
    replaced with the following edges:
    $e_{v_2,A}$, $e_{A,H}$, $e_{H,f_1}$, $e_{f_1,D}$, and $e_{D,v_3}$.
\end{packed_items}

When the road network is augmented, every vertex in $G'$
records the identifiers of the routes
that cover this location.
Figure~\ref{fig:augment} also shows the identifiers of the routes
that are recorded at each vertex in the augmented graph.

In Figure~\ref{fig:augment}, $r_2$ has score $12$, and is attracted by $2$ facilities.
Thus, each edge in $\mathit{G'}$ that is covered by $r_2$
should receive a score $\frac{\mathit{score}(r_2)}{3} = 4$.
Route $r_1$ has score 8, and is attracted by one facility.
Each edge covered by it in $\mathit{G'}$
receives a score $\frac{\mathit{score}(r_1)}{2} = 4$.
Route $r_3$ has score 6, and is attracted by 1 facility.
Therefore, each covered edge received score $\frac{\mathit{score}(r_3)}{2} = 3$.

For each edge in the augmented road network, the algorithm takes an intersection
of the route identifiers of its two vertices, and computes its score.
For instance,
$\mathit{score}({e_{v_2,A}}) = 4+3 = 7$,
$\mathit{score}({e_{A,H}}) = 4+4+3 = 11$.
The scores of other edges can be computed in a similar way.

After that, AUG identifies the optimal edge(s) with the highest score.
Since $\overline{\mathit{AH}}$ has the highest score in $G'$, the optimal edge is $(A,H)$.
It is mapped back, and becomes the segment $\overline{\mathit{AH}}$.
As AUG cannot extend it to a longer segment, $\overline{\mathit{AH}}$ is returned as the result.

\subsection{Analysis}

We analyze the time complexity of the AUG algorithm, and show its completeness and correctness.

\begin{theorem}
The AUG algorithm has time complexity
$O((|E|+|F|+|R|)|R| +|S|)$.
\end{theorem}

\begin{proof}
In AUG,
the graph augmentation takes $O(|F|+2|R|)$ (line~4), because each
route contributes exactly two vertices.

In the first loop (lines 5--9),
the worst case is that facilities and routes are evenly distributed to the
road network so that every edge in $G$ is augmented.
In this case, for an edge $e$, $e.F_c + e.O_c = \frac{|F|+2|R|}{|E|}$.
Therefore, the outer loop takes $|E_F|+|E_R| \leq |E|$ and $|L| = O(\frac{|F|+2|R|}{|E|})$.
So this loop takes time $O(|F|+2|R|)$.

In the second loop (line~11--18), $|E'| \leq |E|+|F|+2|R|$.
In line~15, $|R'| \leq |R|$. Therefore,
this loop takes time $O((|E|+|F|+2|R|)|R|)$.
The third loop takes time complexity $|S|$.
Note that $|S|$ is usually very small.

To summarize, the time complexity of AUG is
$O(|F|+2|R|+(|E|+|F|+2|R|)|R|+|S|)$.
After simplification, the time complexity is
$O((|E|+|F|+|R|)|R| +|S|)$.
\end{proof}

We proceed to show the correctness and completeness of AUG.

\begin{theorem}
A segment output by the AUG algorithm is an optimal segment.
\end{theorem}

\textsc{Proof Sketch.}
In AUG, every edge in the augmented graph is checked to find the the value for $\mathit{optS}$.
AUG then adds a segment iff the segment has a score equal to $\mathit{optS}$.
It implies that any segment in the result set $S$ must is optimal.

\begin{theorem}
The AUG algorithm finds every optimal segment in the graph.
\end{theorem}

\textsc{Proof Sketch.}
This theorem can be proved with the following two points.
First, AUG searches the graph to make sure that every edge is scanned.
Second, if an edge has a score equal to $\mathit{optS}$,
it either is appended to an existing segment in $S$ or is added to $S$ as a new segment that might be extended later.
Therefore, no segment with score equal to $\mathit{optS}$ is missed.

\section{Iterative Partitioning}
\label{sec:ite}

\subsection{Overview}
Although the augmentation approach is effective at finding the optimal segments,
we can improve its efficiency by pruning unpromising segments.

The idea of the ITE algorithm is to quickly identify a subsegment of
an optimal segment (\emph{optimal subsegment}) and then extend the optimal subsegment into an entire optimal segment.
Therefore, ITE organizes the segments using a heap such that
those segments that are most likely to contain an optimal subsegment
get examined first.
If the segment under examination is an optimal subsegment
then the entire optimal segment can be found by extending it.
In addition, the optimal score can be calculated easily.
Otherwise, the segment is partitioned into smaller segments,
whose likelihoods of having an optimal subsegment are also calculated,
upon which they are inserted back into the heap.

Given a segment $s$, we use the scores of the intersecting routes to measure
its likelihood of having an optimal subsegment.
The segment containing the optimal subsegment is likely to
have many intersecting routes, from which it is likely to receive a high score.

For example, in Figure~\ref{fig:ex_1},
initially the edges that intersect any route are inserted into the heap.
The edge $v_1v_3$ has the most intersecting routes and so is likely to contain an optimal segment.
So $v_1v_3$ is partitioned into equal-sized, smaller segments.
ITE calculates the intersecting routes for each of them,
and adds them to the heap.
This process continues until a subsegment of an optimal segment is found.
In this case, a subsegment $s$ of $\overline{\mathit{AH}}$ is found.
Then $s$ is extended to find that $\overline{\mathit{AH}}$ is the entire optimal segment.

Both AUG and ITE partition the edges of the network graph into smaller pieces.
The main difference between ITE and AUG lies in how a subsegment of
the optimal segment is found.
In AUG, the partitioning of edges in the network graph is unguided.
Every edge that has an attracting facility or a route end point is partitioned.
In ITE, the partitioning of edges is guided by the likelihoods
of the edges to have an optimal subsegment.

\subsection{The ITE Algorithm}
Recall that we are interested in finding those segments that contain an optimal subsegment.
Before presenting the ITE algorithm,
we need definitions that relate the score of a segment to the scores of its network points,
as defined in Section~\ref{subsec:interest}.
\begin{definition}
Given a road segment $s$, we define its min score $s.\mathit{min}$
and max score $s.\mathit{max}$ as follows.
\begin{align*}
    s.\mathit{min} & = \min_{p \in s}\mathit{score}_M(p) \\
    s.\mathit{max} & = \max_{p \in s}\mathit{score}_M(p)
\end{align*}
\end{definition}

By definition, an optimal segment $s_{\mathit{opt}}$ has $s_\mathit{opt}.\mathit{min} = s_\mathit{opt}.\mathit{max}$.

Next, we define upper and lower bound scores of a segment
$s$ in order to only process
those segments that may contain an optimal location.

\begin{definition}
\label{def:lb_ub}
Given a segment $s$ and a score distribution model $M$,
let $s.I$ and $s.C$ be defined as in Section~\ref{subsec:routesAndFac},
and let $s.\mathit{lb}$ and $s.\mathit{ub}$ denote the upper and lower bound scores of $s$.
We define:

\begin{align*}
s.\mathit{lb} &=  \sum_{r_i\in s.C}{w_M(j_i, k_i)\mathit{score}(r_i)} \\
s.\mathit{ub} &=  \sum_{r_i\in s.I}{w_M(j_i, k_i)\mathit{score}(r_i)}
\end{align*}
where $s$ is the $j_i$th segment of $r_i$ with $k_i$ attracting facilities, and
$w_M(j_i, k_i)$ computes the fraction of $r_i$'s score to be assigned to $s$ based on $M$.
\end{definition}

If a segment has facilities located on it, the segment has subsegments that may be assigned different scores based on the score distribution model.
In this case, the lower bound score of the segment still takes the smallest score value being assigned to the subsegments,
while the upper bound score takes the largest score value.

\begin{lemma}
Let $M$ be a score distribution model where a route can only distribute non-negative scores.
Let the min and max scores and the lower and upper bound scores of $s$ be defined as above.
Given a road segment $s$, we have $s.\mathit{lb} \leq s.\mathit{min}$ and $s.\mathit{ub} \geq s.\mathit{max}$.
\end{lemma}

\begin{proof}
Suppose a network location $p_1 \in s$ s.t. $\mathit{score}_M(p_1) = s.\mathit{min}$.
Since each route $r \in s.C$ contains $s$, we have $p_1 \in r$.
So $p_1$ at least gains the scores distributed by the routes in $s.C$.
Then $\mathit{score}_M(p_1) \geq \sum_{r_i\in s.C}{w_M(j_i, k_i)\mathit{score}(r_i)}$.

Let $p_2 \in s$ be a location s.t. $\mathit{score}_M(p_1) = s.\mathit{max}$.
We show that the set of routes that contribute scores to $p_2$ is a subset of $s.I$.
The set of routes that contribute scores to $p_2$ consists of two sets,
the set of routes that contain $s$ ($s.C$) and the set of routes that cover $p_2$ ($s.I'$).
Each route in $s.I'$ must also intersect $s$, so $s.I' \subset s.I$.
Since $s.C \subseteq s.I$, we have $(s.C \cup s.I') \subseteq s.I$.
That is, $\mathit{score}_M(p_1) \leq \sum_{r_i\in s.I}{w_M(j_i, k_i)\mathit{score}(r_i)}$.
\end{proof}

Recall that Algorithm PreProcess builds a mapping from each edge $e$ to its intersecting routes $e.R_c$.
We then compute the upper and lower bound scores for a segment $s \subseteq e$ by retrieving its $s.I$ and $s.C$ from $e.R_c$.
The algorithm can use the bounds to prune the segments
that cannot contain an optimal subsegment.
\begin{lemma}
\label{lem:prune1}
Given two segments $s_1$ and $s_2$, if $s_1.\mathit{lb} > s_2.\mathit{ub}$,
then $s_2$ does not contain an optimal subsegment.
\end{lemma}

\begin{proof}
We prove Lemma \ref{lem:prune1} by showing that $s_2$ cannot contain any optimal location.
Assume two points $p_1 \in s_1$ and $p_2 \in s_2$.
We have
$\mathit{score}_M(p_1) \geq s_1.\mathit{lb} > s_2.\mathit{ub} \geq \mathit{score}_M(p_2)$.
So $p_2$ cannot be an optimal location.
\end{proof}

With Lemma~\ref{lem:prune1},
segments that do not contain an optimal subsegment
can be pruned.

The second strategy employed in ITE is to prune the
segments that eventually lead to the same optimal segment.
These segments should be detected and pruned early
to avoid partitioning them further and making unnecessary calculations.

\begin{lemma}
\label{lem:prune2}
Given two segments $s_1$ and $s_2$,
if $s_2.I \subset s_1.C$ and $s_2$ contains an optimal subsegment of an optimal segment,
then $s_1$ also contains an optimal subsegment of the same optimal segment.
\end{lemma}

\begin{proof}
Let the optimal segment be $s_{\mathit{opt}}$,
and let $s_2$ contains an optimal subsegment of $s_{\mathit{opt}}$.
Then we have $s_{\mathit{opt}}.C \subseteq s_2.I$ because
every route that contain the optimal subsegment must intersect with $s_2$.

Since $s_2.I \subseteq s_1.C$, we have $s_{\mathit{opt}}.C \subseteq s_1.C$.
By the definitions of segment score and optimality,
we also have $s_{\mathit{opt}}.C = s_1.C$.
Therefore, by the definition of segment score,
$s_1$ is also an optimal subsegment of $s_{\mathit{opt}}$.
\end{proof}

Once the result set is not empty,
Lemma~\ref{lem:prune2} allows us to prune segments
that lead to the same optimal segment.
We study the effectiveness of the pruning strategies in the experimental evaluation.

Figure~\ref{fig:ub_lb} shows the edge $e_{2,3}$ from Figure~\ref{fig:ex_1}.
It illustrates the calculation of segment score upper bound and lower bound.
The edge $e_{2,3}$ has a facility $f_1$ built on it,
resulting in two subsegments, $s_1$ from the beginning to $f_1$ and $s_2$ from $f_1$ to the end.
The routes $r_1$ and $r_3$ are attracted by one facility,
whereas $r_2$ is attracted by two facilities.
We show how to calculate the score upper and lower bounds for both $s_1$ and $s_2$.
Segment $s_1$ is intersected
by $r_1$ and $r_3$, and contained by $r_2$.
Therefore,
$s_1.\mathit{lb} = \frac{1}{3}{\mathit{score}(r_2)}$,
$s_1.\mathit{ub} = \frac{1}{2}\mathit{score}(r_1)+ \frac{1}{3}\mathit{score}(r_2) + \frac{1}{2}\mathit{score}(r_3)$

\begin{figure}[ht]
  \centering
    \includegraphics[width=0.6\textwidth]{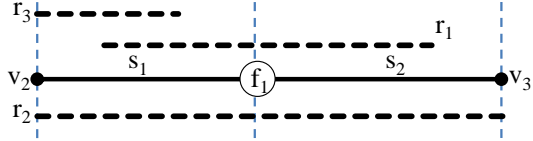}\hfill
  \caption{Segment Upper and Lower Bound}
  \label{fig:ub_lb}
\end{figure}

Similarly, $s_2$ is intersected by $r_1$,
and contained by $r_2$. Therefore,
$s_2.\mathit{lb} = \frac{1}{3}\mathit{score}(r_2)$,
$_2.\mathit{ub} = \frac{1}{2}\mathit{score}(r_1)+ \frac{1}{3}\mathit{score}(r_2)$

Algorithm~\ref{alg:findSeg} shows the pseudo-code of the ITE algorithm.
This algorithm uses a priority queue $Q$
that is sorted on the upper bound score of every segment.
A variable called $\mathit{maxLb}$ is used to keep track
of the maximum lower bound score seen so far.

\begin{algorithm}[ht]
\SetKwData{Left}{left} \SetKwData{This}{this} \SetKwData{Up}{up}
\SetKwFunction{Union}{Union}
\SetKwFunction{FindCompress}{FindCompress} \SetKwInOut{Input}{input}
\SetKwInOut{Output}{output}
\SetKwIF{If}{ElseIf}{Else}{if}{then}{else if}{else}{endif}
\SetKwFor{For}{for}{do}{endfor} \SetKwFor{While}{while}{do}{endw}
\SetKwFor{ForEach}{foreach}{do}{endfch}
\caption{ITE($G, R, F, \delta, \beta, M$) }
\label{alg:findSeg}
Init. $e\in G.E$ s.t. $e.\mathit{lb} \leftarrow 0$ and $e.\mathit{ub} \leftarrow \sum_{r_i\in e.R_c}{w_M(j_i, k_i)\mathit{score}(r_i)}$;  \\
$S \leftarrow \emptyset$;
$Q.\mathit{enqueue}(G.E)$;
$\mathit{maxLb} \leftarrow 0$; \\
\While{$Q \neq \emptyset$}{
    $\mathit{currSeg} \leftarrow Q.\mathit{dequeue}$; \\
    $\mathit{split} \leftarrow \mathit{false}$; \\
    \If{$\mathit{currSeg}.\mathit{ub} > \mathit{maxLb}$ }{
        $\mathit{split} \leftarrow \mathit{true}$;
    }\ElseIf{$\mathit{currSeg}.\mathit{ub} = \mathit{maxLb}$ }{
        \If{$\mathit{currSeg}.\mathit{ub} = \mathit{currSeg}.\mathit{lb}$}{
            $S.\mathit{add}.(\mathit{currSeg})$;
        }\ElseIf{$\nexists s \in S$ such that $\mathit{currSeg}.I \subseteq s.C$}{
            $\mathit{split} \leftarrow \mathit{true}$;
        }
    }
    \If{$\mathit{split}$}{
         $\mathit{ss} \leftarrow \mathrm{SplitSegment}(\mathit{currSeg}, \beta)$;
    }
    \ForEach{$s \in \mathit{ss}$}{
        \ForEach{$\mathit{r} \in \mathit{currSeg}.I$}{
            \If{$\mathrm{intersects}(\mathit{r}, s)$}{
                $s.I \leftarrow s.I.\mathrm{add}(\mathit{r}\})$; \\
                $s.\mathit{ub} \leftarrow s.\mathit{ub} + w_M(j, k)\mathit{score}(\mathit{r})$; \\
            }
            \If{$\mathrm{contains}(\mathit{r}, s)$ }{
                $s.C \leftarrow e.C.\mathrm{add}(\mathit{r})$; \\
                $s.\mathit{lb} \leftarrow s.\mathit{lb} + w_M(j, k)\mathit{score}(\mathit{r})$;
             }
        }
        \If{$s.\mathit{lb} > \mathit{maxLb}$}{
           $\mathit{maxLb} \leftarrow s.\mathit{lb}$;
        }
         $Q.\mathit{enqueue}(s)$;
    }
}
\ForEach{$s \in S$}{
    Find the entire optimal segment of $s$ by overlapping
    the route usage objects $\mathit{r} \in s.C$ one by one.
}
\end{algorithm}

First, ITE initializes the edges such that each has a lower bound
score $0$ and an upper bound score computed as in Definition~\ref{def:lb_ub}
(line~1). It also initializes the result set $S$, enqueues the edges
$G.E$ of the road network graph, and initializes variable
$\mathit{maxLb}$ (line~2). It then enters the loop and pops out the
top element from $Q$ (lines~3--4). The flag variable $\mathit{split}$,
indicating whether or not the current segment needs to be partitioned,
is set to false at the beginning of each iteration (line~5). Next, if
the upper bound score of $\mathit{currSeg}$ exceeds $\mathit{maxLb}$
then it needs to be further partitioned, so $\mathit{split}$ is set to
true (lines~6--7).  If the upper bound score of $\mathit{currSeg}$ is
equal to $\mathit{maxLb}$, we have found a result segment if the upper
and lower bound scores are the same. Then $\mathit{currSeg}$ is added
to the result set (lines~8--10). However, if the upper and lower bound
scores differ, ITE tests whether $\mathit{currSeg}$ might lead to an
optimal subsegment of a new optimal segment that is not seen before.
Then ITE checks if there is a result $s$ in $S$ such that $s.C$ is
subset of $\mathit{currSeg}.I$ (see Lemma~\ref{lem:prune2}).  If no,
$\mathit{split}$ is set to true (lines~11--12).

If $\mathit{split}$ is true, the function partitions $\mathit{currSeg}$
into subsegments with the procedure $\mathrm{SplitSegment}$,
which partitions a segment $G$ into $\beta$ equal length subsegments (lines~13--14).
Here $\beta$ is a tunable parameter.
In the experimental studies, we show the effect of $\beta$.

The intersection set, contain set, and lower and upper bound scores for each segment
output by $\mathrm{SplitSegment}$ are computed and inserted into $Q$ (lines~15--22). Next, $\mathit{maxLb}$ is updated if the subsegment has a higher lower bound score (lines~23--24).
Then these subsegments are added back to $Q$ (line~25).

Upon exiting the loop, each optimal segment
is extended to its full length
by overlapping the routes that contribute scores to the segment (lines~26--27).

\begin{figure}[ht]
\centering
  \subfigure[]
{\label{fig:ite_part1}\includegraphics[width=0.48\textwidth]{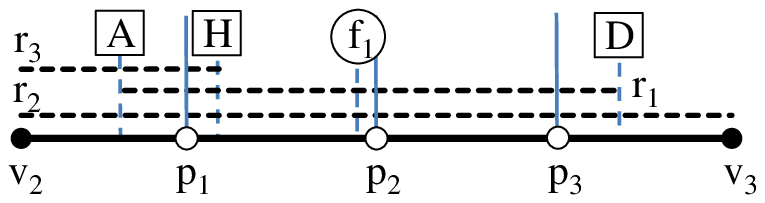}}\hfill
  \subfigure[]
  {\label{fig:ite_part2}\includegraphics[width=0.48\textwidth]{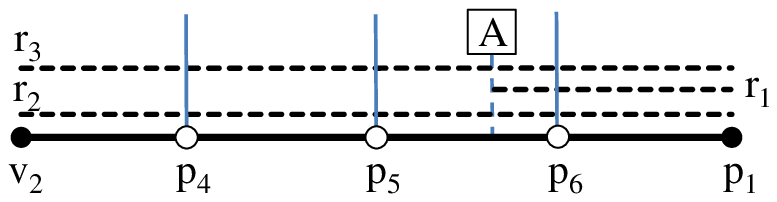}}\hfill
  \subfigure[]
  {\label{fig:ite_part3}\includegraphics[width=0.48\textwidth]{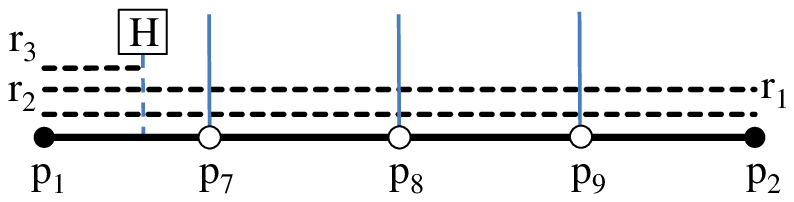}}\hfill
  \subfigure[]
  {\label{fig:ite_part4}\includegraphics[width=0.48\textwidth]{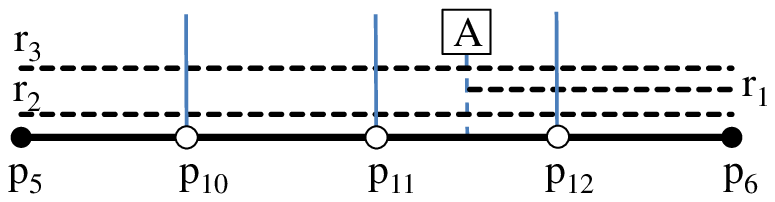}}\hfill
  \caption{ITE Execution Example}
   \label{fig:ite_part}
\end{figure}

We continue to use Figure~\ref{fig:ex_1} to illustrate the execution of
Algorithm~\ref{alg:findSeg}.
We show the iterative partitioning of $e_{2, 3}$ in Figure~\ref{fig:ite_part}.
Table~\ref{tbl:execution} shows the top entries of the queue obtained from partitioning $e_{2,3}$,
together with their upper and lower bound scores during the execution of ITE.
Double lines separate iterations. The segment at the top of the queue is in bold.

\begin{table}[ht]
\begin{center}
   \begin{tabular}{c|c|c||c|c|c}
    \hline
    Segment & $ub$ & $lb$ & Segment & $ub$ & $lb$\\ \hline
    \textbf{$\overline{v_2p_1}$} & \textbf{11} & \textbf{7} & $\overline{p_2p_3}$ & 8 & 8 \\ \hline
    $\overline{p_1p_2}$ & 11 & 8 & $\overline{p_3v_3}$ & 8 & 4   \\ \hline
     \hline
     \textbf{$\overline{p_1p_2}$} & \textbf{11} & \textbf{8} & $\overline{p_2p_3}$ & 8 & 8\\ \hline
     $\overline{p_5p_6}$ & 11 & 8 & $\overline{p_3v_3}$ & 8 & 4 \\ \hline
     $\overline{p_6p_1}$ & 11 & 11 & $\cdots$ \\ \hline
     \hline
     \textbf{$\overline{p_5p_6}$} & \textbf{11} & \textbf{8} & $\overline{p_2p_3}$ & 8 & 8 \\ \hline
     $\overline{p_6p_1}$ & 11 & 11 & $\overline{p_3v_3}$ & 8 & 4 \\ \hline
     $\overline{p_1p_7}$ & 11 & 8 & $\cdots$ \\ \hline
     \hline
     \textbf{$\overline{p_6p_1}$} & \textbf{11} & \textbf{11} & $\overline{p_{12}p_6}$ & 11 & 11 \\ \hline
     $\overline{p_1p_7}$ & 11 & 8 & $\overline{p_2p_3}$ & 8 & 8 \\ \hline
     $\overline{p_{11}p_{12}}$ & 11 & 7 & $\cdots$ \\ \hline
    \end{tabular}
    \end{center}
    \vspace{-0.4cm}
    \caption{ITE Execution Example}
  \label{tbl:execution}
\end{table}

Below, we calculate the upper and lower bound scores of
$\overline{v_2p_1}$, which is intersected with $r_1$, $r_2$, and $r_3$,
and contained by $r_2$ and $r_3$.
Therefore, $\mathit{ub}(\overline{v_2p_1})=\sum_{i=1}^{3}{\mathit{score}(r_i)}{k_i}=\frac{8}{2} + \frac{12}{3} + \frac{6}{2}=11$
and $\mathit{lb}(\overline{v_2p_1}) = \sum_{i=2}^{3}{\mathit{score}(r_i)}{k_i}=\frac{12}{3} + \frac{6}{2}=7$.
The upper and lower bound scores of other segments can be computed in a similar way.

Since segment $\overline{v_2p_1}$ has the largest upper bound score and
its upper bound is not the same as the lower bound,
it is split as shown in Figure~\ref{fig:ite_part2}.
The upper and lower bound scores of the subsegments are also computed.
Now $\mathit{maxLb} = 11$

In the next iteration, segment $\overline{p_1p_2}$
has the largest upper bound score.
But still, its upper bound is not the same as its lower bound.
It is then split (Figure~\ref{fig:ite_part3}).
The upper and lower bound scores of
the subsegments are also computed, and
$\mathit{maxLb} = 11$.

The next segment under examination is $\overline{p_5p_6}$,
which is split again because its upper bound is different from
its lower bound (Figure~\ref{fig:ite_part4}). Still, $\mathit{maxLb} = 11$.

The next segment under examination is $\overline{p_6p_1}$,
whose upper and lower bounds are the same.
The upper bound of $\overline{p_6p_1}$ is also same
as $\mathit{maxLb}$.
Therefore, $\overline{p_6p_1}$ is added into the result
set as an optimal subsegment.

The process continues until an optimal subsegment of
every optimal segment in the network graph is found.
$Q$ is updated at the end of each iteration.
Note that ITE does not need to examine
those segments with score upper bound less than
$\mathit{maxLb}$ (11 in this case), resulting
in a substantial reduction of the search space.

In the end, the entire
optimal segment $\overline{\mathit{AH}}$ can be found by overlapping the
routes $\overline{\mathit{AH}}.C$, $r_1$, $r_2$, and $r_3$.

\subsection{Analysis}
\label{sec:theoretical_analysis}
We consider the correctness and completeness of ITE and analyze its time complexity.

The correctness of ITE depends on finding the subsegments
with the maximum score correctly. Here, we prove that
the algorithm terminates and returns subsegments that have
the maximum score. We must show that after a
finite number of iterations, ITE produces a subsegment $s$ such that $s.\mathit{ub} = s.\mathit{lb}$ where
$s.\mathit{ub}$ is the maximum score among all the subsegments.
First, we know that when $s.\mathit{ub} = s.\mathit{lb}$, $s$ is a consistent segment with the score $s.\mathit{ub}$.
Since $s.\mathit{ub}$ is the maximum among all the subsegments
ensured by the property of the priority queue, $s$ is a subsegment with the
maximum score.
Since ITE always examines the subsegment with the maximal $s.\mathit{ub}$,
we only need to show that ITE terminates.
This can be shown by the following
properties. (1) The maximum $\mathit{ub}$ value decreases and  (2) The maximum $\mathit{lb}$ increases.
(3) The maximum $\mathit{ub}$ and $\mathit{lb}$ values converge to the same value after a number of iterations.

Next, to prove completeness, we show that for each optimal
segment, ITE is able to find a subsegment with the maximum score
that is contained within the optimal segment. Let $s_i$ and $s_j$
be the subsegment of two distinct optimal segments.
Without loss of generality, suppose ITE has found $s_i$.
We show that ITE also finds $s_j$ instead of pruning it.
Recall that ITE uses two pruning criteria to prune a
subsegment.
The first criterion says that $s_j$ can be pruned if $s_i.\mathit{lb} > s_j.\mathit{ub}$.
Since both $s_i$ and $s_j$ are subsegments of optimal segments with the same optimal score $\mathrm{OPT}$,
we have $s_j.\mathit{ub} \geq \mathrm{OPT} \geq s_i.\mathit{lb}$. Therefore, this pruning criterion does not apply.
The second criterion states that $s_j$ can be pruned if $s_j.I \subseteq s_i.C$.
Since $s_i$ and $s_j$ are subsegments of different optimal segments, $s_j.C \nsubseteq s_i.C$.
We also know that $s_j$ is a subsegment with the maximum score, hence $s_j.I = s_j.C$.
Putting them together, we have $s_j.I = s_j.C \nsubseteq s_i.C$.
Hence the second pruning criterion also do not apply.
Thus, ITE does not prune $s_j$, but detects it as a part of an entire segment which is also found.
Therefore, ITE will find all the optimal segments.

\begin{theorem}
The time complexity of ITE is $O((\mathit{log}|R|+|S|)|R|)$.
\end{theorem}

\begin{proof}
%

In the priority queue operations, ITE iteratively splits
the segment with the maximal score upper bound.
The number of splits corresponds to the height of the tree with fan-out $\beta$.
If $\beta=4$, we get a quadtree.
According to~\cite{FGP93}, the asymptotic height of the quadtree is $\mathit{log}|R|$.
For each subsegment $s$, ITE uses a loop to find its intersection set $s.I$ and contain set $s.C$.
We have $\mathit{ss}.I \subseteq R$, so the time complexity of the loop is $O(|R|)$.
Thus, time complexity of the while loop is $O(|R|\mathit{log}|R|)$.

The second loop depends on the size of the result set $S$.
Since $s.C \subseteq R$, the time complexity of this loop is
$O(|R||S|)$.

In total, the time complexity of ITE is $O((\mathit{log}|R|+|S|)|R|)$.
\end{proof}

\section{Experimental Study}
\label{sec:expt}
This section reports on empirical studies that aim to elicit design properties of the proposed framework and,
in particular, of the AUG and ITE algorithms.
The studies use a real spatial network and real facility and trajectory data, as well as synthetic data.

The experiments covered in this section were performed on an Intel Xeon (2.66Ghz) quad-core machine with 8 GB of main memory running Linux (kernel version 2.6.18).
Both of the algorithms were implemented in Java.
Every instantiation of JVM was allocated 2 GB of virtual memory.
We first describe the data used in the experiment as well as the parameter settings.
Then we cover experiments that target different aspects of the algorithms.

\subsection{Data Sets and Parameter Settings}
\subsubsection{Road Network}

The digital road network TOP10DK~\footnote{http://tinyurl.com/bqtgh2g} was used for our experiments.
It contains all of Denmark at a fine granularity.

To construct the road network graph, we first identify the vertices.
An edge exists between two vertices $v_1$ and $v_2$ as long as there exists a road segment connecting $v_1$ and $v_2$.
In total, the graph contains 465,057 vertices and 920,218 edges.

In order to study the performance of the algorithms thoroughly,
we used a real-world data set and a synthetic data set.
Each data set contains a collection of routes (GPS recordings received from drivers) and a collection of facilities.
Both data sets share the same underlying road network.

\subsubsection{Route Data Preparation}
We obtained the real route traversal data set from the "Pay as You Speed" project \cite{LAT07}~\footnote{http://www.trafikdage.dk/td/papers/papers07/tdpaper27.pdf}.
The data set is obtained from vehicles driving in North Jutland, Denmark.
The data set contains 39,688,695 GPS points
produced by 151 different drivers in the period from October 1, 2007 to January 31, 2008.
In this data set, each route is represented by a sequence of GPS points that may deviate from the underlying road network.
To solve this problem, we use an existing technique
by Tradi\v{s}auskas et al.~\cite{TJL07} to map-match the route data onto the underlying road network.
Then the sequence of traversed edges has also to be determined because two consecutive GPS points may
be matched to different edges.
To achieve it, we use a bidirectional Dijkstra's algorithm provided by Pohl~\cite{pohl1971}.

In addition, stationary points, when reported GPS locations are the same for consecutive time points for the same user, are removed.
Further, different trips of users were identified from the set of GPS recordings.
We distinguish a new route when the time period between two consecutive GPS points is more than 3 minutes.
In total, we obtain 51,146 routes.
The median number of GPS points of the routes is 488.
The median length of the routes in the real data sets is 6524.81.

We generate synthetic routes by simulating the movement of a vehicle that emits GPS points with a fixed frequency (e.g., 0.1 Hz).
The length of the each route is thus the speed of the car times the number of GPS points  it emits.
In the simulation, routes are allowed to have variable lengths.
So when starting a new synthetic route, we first generate a random number between 480 and 520 for the number of GPS points.
We use 480 and 520 because the median number of GPS points of the routes in the real data set is 488.
Then we randomly select a network point to start a new route.
When taking the next point, we follow the graph and traverse to
the next edge (randomly pick one if more than one outgoing edge exists).
The sampling frequency is fixed for one data set to simulate a real life application.
We then vary the sampling frequency, resulting in three different data sets, i.e.,
short, medium, and long, with the median lengths of routes being 3405.76, 8030.42, and 12890.12, respectively.

In both the real and synthetic data sets, for each route, we use a random number generator to generate the user count and route usage randomly from 1 to 20.

\subsubsection{Facilities}
The facility data set contains 16,577 places of interest located throughout Denmark.
The exact address of each facility can be looked up from yellow pages.
Since it is meaningless to take businesses of different types,
we group the facilities according to their types (e.g., fast food, salon, supermarket).
In all the experiments below, the facilities are of the same type.
When generating the synthetic facility data set, we randomly pick network points from the network.
Every facility in either the real data set or the synthetic attracts at least one route.

Statistics on the data sets and the settings for key parameters are summarized in Table~\ref{tbl:setting_real}.
The default values are in bold.

\subsubsection{Scoring Function and Score Distribution Model}
We observe from the experiments that
the scoring function and the score distribution model do not affect the performance of the two algorithms.
Therefore, we only show the experimental results produced when using the first scoring function and the first proposed model.

\begin{table}[h]
\label{tbl:setting}
\begin{center}
    \begin{tabular}{ | c | c |}
    \hline
    \textbf{Parameter} & \textbf{Range} \\ \hline
    $\delta$ & 0.02, 0.04, \textbf{0.06}, $\ldots$, 0.12 \\ \hline
    $\beta$ & 2, 3, \textbf{4}, 5, 6  \\ \hline
    Num Routes in Real Data & 5k, 10k, $\ldots$, \textbf{25k} \\ \hline
    Num Routes in Synthetic Data & 10k, 15k, $\ldots$, \textbf{30k} \\ \hline
    Num Facilities & 600, 800, \textbf{1k}, $\ldots$, 1.4k\\ \hline
    \end{tabular}
    \end{center}
    \vspace{-0.3cm}
    \caption{Experimental Settings}
  \label{tbl:setting_real}
\end{table}

\subsection{Effect of $\delta$}
Recall that a facility attracts a route if their distance is no further than $\delta$.
Figure~\ref{fig:delta} shows the performance and optimal scores when varying $\delta$ on real data.

\begin{figure}[ht]
\vspace{-0.2cm}
  \centering
  \subfigure[Performance vs $\delta$, Real ]
  {\label{fig:run_time_delta}\includegraphics[width=0.45\textwidth]{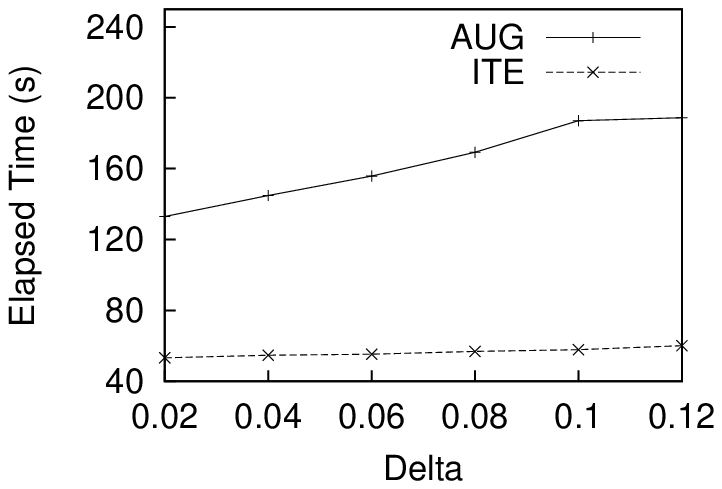}}\hfill
  \subfigure[Score vs $\delta$, Real] {\label{fig:scoreDelta}\includegraphics[width=0.45\textwidth]{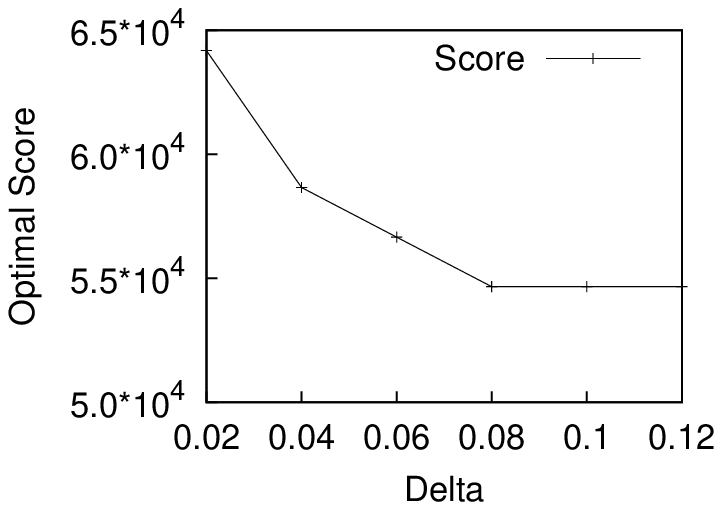}}\hfill
\vspace{-0.2cm}
  \caption{Effect of $\delta$}
    \label{fig:delta}
\vspace{-0.1cm}
\end{figure}
Although the running times for both algorithms increase when $\delta$ increases (Figure~\ref{fig:run_time_delta}),
the two algorithms exhibit different patterns.
When $\delta$ increases from
0.02 to 0.1, AUG increases much faster than ITE.
AUG has to explore further on the edges to find the attracting
facilities for each route traversal, in order to decide whether to include them in the augmented graph.
This may be the reason why AUG increases more rapidly than ITE.
When $\delta$ increases from 0.1 to 0.12, the running time of AUG increases slower.
The reason may be that less facilities are taken into account.
In reality, some facilities prefer locations near the junctions,
so the density of facilities in the middle of roads might be less.
The increase of the running time of ITE is less, and there is no
sudden change, indicating that $\delta$ has little effect on ITE.

Since the optimal scores output by both algorithms are
the same, we plot one figure to
show the effect of $\delta$ (Figure~\ref{fig:scoreDelta}).
The optimal scores decrease like a staircase.
The reason is that increasing $\delta$ may increase the number of
attracting facilities for a route, resulting in decreased scores of segments received from
the routes according to the score distribution model.

\subsection{Effect of $\beta$}
Recall that $\beta$ is the number of subsegments produced when a segment is partitioned.
It is a user-specified parameter.
Figure~\ref{fig:beta} shows the effect of $\beta$ on the running time of ITE.

\begin{figure}[ht]
  \centering
    \includegraphics[width=0.45\textwidth]{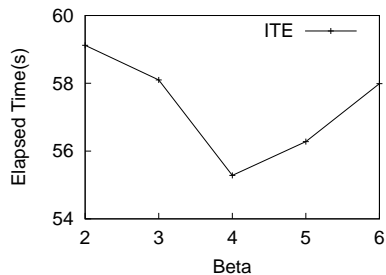}
\vspace{-0.2cm}
  \caption{Effect of $\beta$, Real}
  \label{fig:beta}
\vspace{-0.1cm}
\end{figure}

Initially, as $\beta$ value increases, the running time decreases.
However, beyond a certain $\beta$ value (4 in the figure), with further increase in the next value, the running time starts to increase.
The best performance of ITE occurs when $\beta=4$.
When the $\beta$ value is smaller than 4, the ``zooming-into" an optimal subsegment may not be as fast as when $\beta=4$.
On the other hand, when the $\beta$ value is greater than 4, computing the lower and upper
bounds of the subsegments can take substantial time, and thus the increase in running time.

\subsection{Effect of the Number of Routes}
Figure~\ref{fig:perf_num_routes} shows the performance when varying the number of routes using real and synthetic data.
Algorithms AUG and ITE perform equally well for a small number (5k) of routes.
But the running time of AUG grows much more rapidly than that of ITE with the increase of the quantity of routes.
This is expected from the  time complexity analysis of AUG and ITE.

\begin{figure}[ht]
  \subfigure[Performance vs. Num Routes, Real]
  {\label{fig:perf_routes_real}\includegraphics[width=0.45\textwidth]{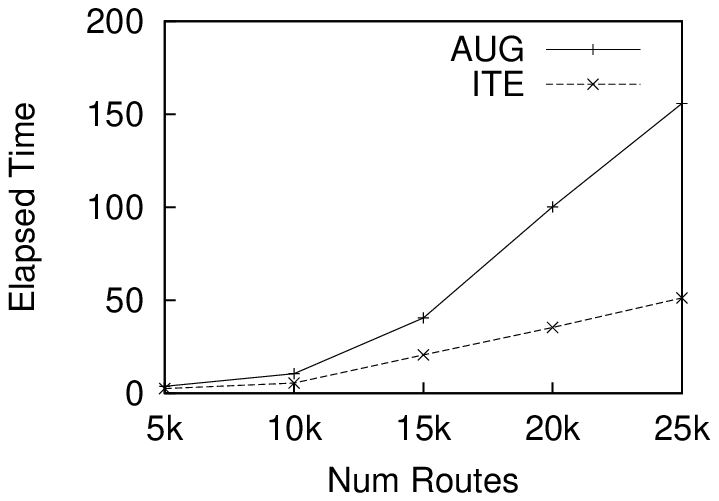}}\hfill
  \subfigure[Performance vs. Num Routes, Synthetic]
  {\label{fig:perf_routes_syn}\includegraphics[width=0.45\textwidth]{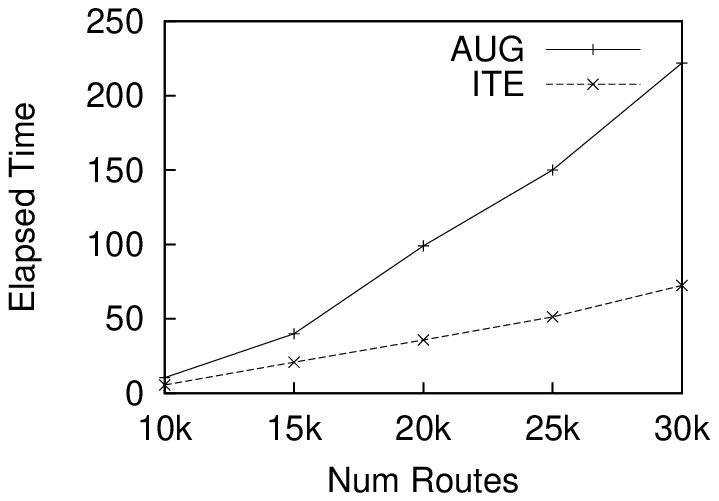}}\hfill
  \vspace{-0.2cm}
  \caption{Effect of the Number of Routes on Performance}
   \label{fig:perf_num_routes}
\vspace{-0.1cm}
\end{figure}

\subsection{Effect of Route Length}

In this set of experiments, we study the effect of the length of routes
on the performances of both algorithms.
The three data sets used in the experiments are explained above.
Figure~\ref{fig:route_length} shows the results.
For both algorithms, more time is needed for longer routes
when the number of route traversals ranges from 5k to 25k.
Again, the running time of AUG increases faster than that of ITE.

\begin{figure}[ht]
  \centering
    \includegraphics[width=0.45\textwidth]{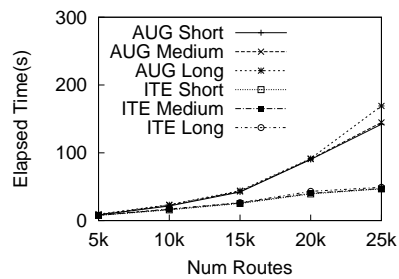}
\vspace{-0.2cm}
  \caption{Effect of Route Length}
  \label{fig:route_length}
\vspace{-0.1cm}
\end{figure}

For AUG, computing the Attraction List for the
vertices takes longer time as each route covers more vertices on average.
For ITE, each edge intersects more routes on average.
So after splitting a segment, more routes have to be
examined to calculate the lower and upper bound scores
of the subsegments, resulting in longer running time.

\subsection{Effect of the Number of Facilities}

Figure~\ref{fig:num_f} shows the running time of AUG and ITE when
varying the number of facilities.

\begin{figure}[ht]
  \subfigure[Performance vs. Num Facilities, Real]
  {\label{fig:perf_fac_real}\includegraphics[width=0.45\textwidth]{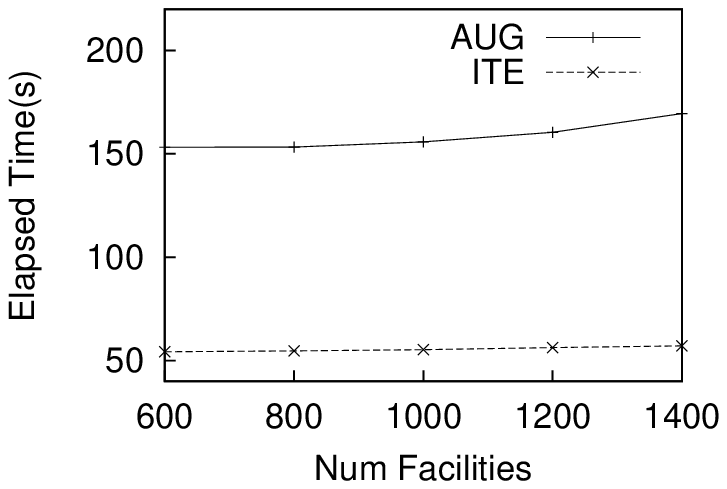}}\hfill
  \subfigure[Performance vs. Num Facilities, Synthetic]
  {\label{fig:perf_fac_syn}\includegraphics[width=0.45\textwidth]{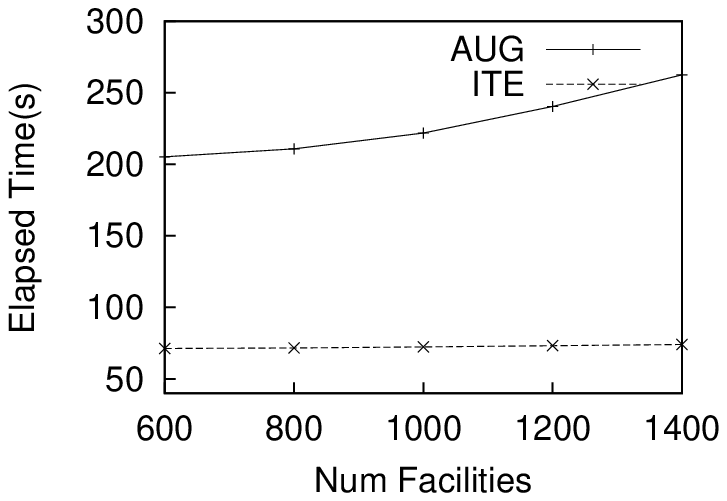}}\hfill
  \vspace{-0.2cm}
  \caption{Effect of the Number of Facilities}
  \label{fig:num_f}
\vspace{-0.1cm}
\end{figure}

For both kinds of routes, AUG is affected slightly more by the
increase in facilities.  The reason is that in AUG, facilities have to
be augmented, and then attraction lists have to be calculated for
them, resulting in substantial computation. In contrast, facilities
cause little computation in ITE. When ITE partitions the segments, no
house-keeping is necessary for facilities. It just needs to adjust the
relative positions if the facilities according to the newly produced
segments.

\subsection{Effectiveness of Pruning Strategies}
In this set of experiments, we study the effectiveness of the
pruning strategies in ITE by keeping track of the number
of segments generated, partitioned, and pruned in the course
of finding the optimal segments.
Figure~\ref{fig:num_prune} shows the respective segments
generated, split, and pruned by Lemma~\ref{lem:prune1} and
Lemma~\ref{lem:prune2} when running ITE with default settings.
Label ``total'' means the total number of generated subsegments,
``splits" is the number of subsegments that needs further splitting,
``prune1" is the number of subsegments that are pruned using Lemma~\ref{lem:prune1}, and ``prune2" is the number of subsegments that
are pruned using Lemma~\ref{lem:prune2}.

\begin{figure}[ht]
\vspace{-0.2cm}
  \subfigure[Real]
  {\label{fig:prune_real}\includegraphics[width=0.45\textwidth]{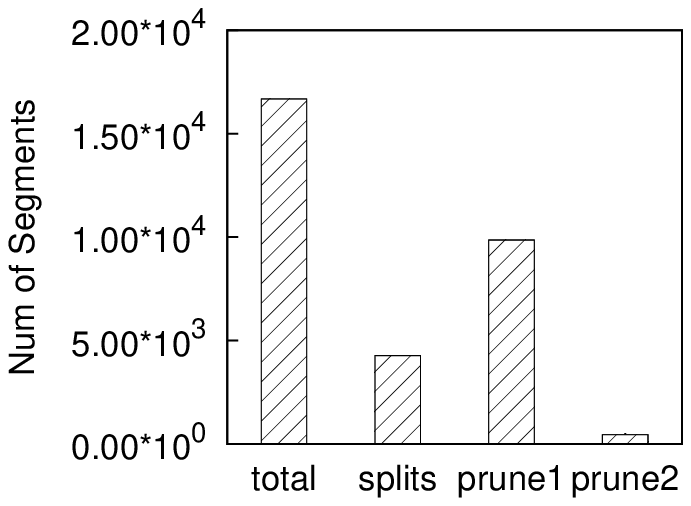}}\hfill
  \subfigure[Synthetic]
{\label{fig:prune_syn}\includegraphics[width=0.45\textwidth]{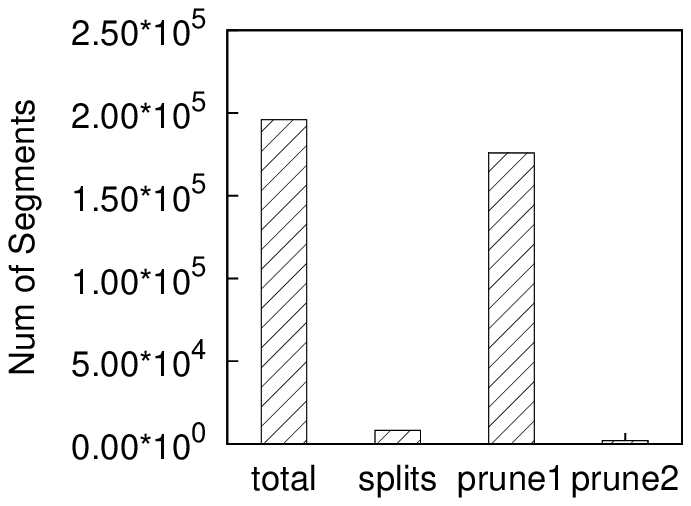}}\hfill
\vspace{-0.2cm}
  \caption{Effect of Pruning Strategies}
    \label{fig:num_prune}
\vspace{-0.2cm}
\end{figure}

It is observed that in the real data set the number of segments that require further splitting is 4,273, which is approximately
25\% of the number of total segments, whereas the number is between 5\% and 10\% in the synthetic data set.

In the real data set, almost 60\% of the generated segments
that cannot contain an optimal segment are been
pruned by Lemma~\ref{lem:prune1}.
In contrast, Lemma~\ref{lem:prune2} prunes 2.7\% of
the total segments.

In the synthetic data set where route traversals are generated more evenly
throughout the entire map, Lemma~\ref{lem:prune1} prunes almost
10\% of the total generated segments.
Lemma~\ref{lem:prune2} prunes around 1\% of the total segments.

\section{Related Work}
\label{sec:related}

The paper's study relates to two previously studied problems, the
facility location problem~(FLP) and flow intercepting facility
location problem~(FIFLP), which we cover in turn.

\subsection{Facility Location Problem}
The classical facility location problem~\cite{CDL05,FH09,M01,NP05,F06}
takes as input a finite set $C$ of customer locations and a finite set
$P$ of candidate facility locations, and it returns $k$ ($k > 0$)
facilities in $P$ that optimizes a predefined metric.

The \emph{single facility location} problem~\cite{FH09,NP05} finds one
location in $P$ that optimizes a predefined metric with respect to a
set $C$ of customer locations. It assumes that no facility has been
built previously; in contrast, our optimal segment problem permits the
presence of a set $F$ of existing facilities.

The \emph{online facility location} problem~\cite{F06,M01} assumes a
dynamic setting, where (i) the set $C$ of customers is initially
empty, and (ii) new customers may be inserted into $C$ as time
evolves.  The solution to this problem constructs facilities one at a
time, such that its quality (with respect to some metric) is
competitive in comparison to solutions that are given all customer
points in advance. This problem assumes that the set $P$ of candidate
facility locations is finite, while our optimal segment problem does
not.

Many works~\cite{DZX05,WOY09,XYL10,ZDX06,ZWL10,CDL05} study another
variant of the facility location problem, the so-called the
\emph{optimal location} (OL) \emph{problem}, where only the optimal
locations are returned from an \emph{infinite} number of candidate
locations, given a finite set of preexisting facilities $F$. The
problem is studied in $L_p$ space.  Recently, Xiao et al.~\cite{XYL10}
extends the problem to a spatial network setting, using network
distance in place of $L_p$ distance. Our optimal segment problem is
related to the OL query, but uses route traversals instead of static
customer point locations. The techniques presented in these previous
works cannot be applied to solve the optimal segment problem.

\vspace{-0.2cm}
\subsection{Flow Intercepting Facility Location}
The \emph{flow intercepting facility location} (FIFL) \emph{ problem}
is similar to our problem in that it models demand by means of
customer flows. Here, customer trips are pre-planned, and customers
can choose to visit a facility or not during their trips by deviating
from a pre-planned route.

Hodgson~\cite{H81,H90} was the first to identify and study an
FIFL-type problem where the placement of facilities minimizes the
total deviation from preplanned trips made by a population of
customers.
Later, Berman and collaborators investigate a variety of versions of
this problem: (i) the optimal location for discretionary
facilities~\cite{BBL95}, (ii) facility location given probabilistic
flows~\cite{BKX95}, (iii) locating facilities with finite
capacities~\cite{B95}, (iv) locating facilities when the level of
customer usage of a service depends on the number of facilities they
encounter along their path~\cite{AB96}, (v) locating competitive
facilities (demand and flow coverage problem)~\cite{B97}.

Our study differs from this existing work in important ways. We assume
a realistic setting and propose efficient means of placing a facility
on a road segment, considering existing facilities and customer
movements derived from GPS data. Our framework enables the use of
scoring functions that generate scores from customer traversals of
routes, and it enables the use of models that distribute these scores
to road segments. The framework is open to such functions and models
and thus enables the modeling of a wide variety of scenarios.  Our
approach can easily be augmented to model the unavailability of
locations in a spatial network, so that such locations are not
considered in results.

\section{Conclusions and Future Work}
\label{sec:conclusion}

The paper formalizes a modern version of the classical facility
location problem that takes into account the availability of customer
trajectory data that is constrained to a road network, rather than
simply assuming the availability of static customer locations. In the
resulting framework, route traversals by customers rather than
customer locations are attracted by facilities. The framework enables
a wide variety of choices for assigning scores to the routes traversed
by customers and for distributing these scores to segments in the
underlying road network, thus offering flexibility that aims to enable
applications with different types of facilities. We believe that this
work provides a new and realistic generalization of the classical
facility location problem.

Two algorithms, AUG and ITE, are provided to solve this generalized
problem. AUG takes a graph augmentation approach, and ITE iteratively
partitions road segments into smaller pieces (subsegments) while using
a scoring mechanism to guide the selection of promising segments for
further partitioning.
The paper reports on empirical studies with both real and synthetic
routes map-matched to a real spatial network that demonstrate
practicality of the proposed algorithms.  Algorithm ITE outperforms
AUG thanks to its sophisticated pruning techniques hat effectively
reduce the search space.

Several interesting directions for future work exist, including the
following two. First, the optimal segments can be incrementally
evaluated when new routes are available. Incremental evaluation allows
more flexibility when new routes are continuously added and may help
improve the performance. Second, future work may consider finding
top-$k$ segments.

\bibliographystyle{abbrv}

\begin{thebibliography}{10}

\bibitem{ABK09}
R.~Aboolian, O.~Berman, and D.~Krass.
\newblock Efficient solution approaches for a discrete multi-facility
  competitive interaction model.
\newblock {\em Ann. Oper. Res.}, 167(1):297--306, 2009.

\bibitem{AB96}
I.~Averbakh and O.~Berman.
\newblock Locating flow-capturing units on a network with multi-counting and
  diminishing returns to scale.
\newblock {\em EJOR}, 99(3):495--506, 1996.

\bibitem{B95}
O.~Berman.
\newblock The maximizing market size discretionary facility location problem
  with congestion.
\newblock {\em Socio-Economic Planning Sciences}, 29(1):39--46, 1995.

\bibitem{B97}
O.~Berman.
\newblock Deterministic flow-demand location problems.
\newblock {\em J. Oper. Res. Soc.}, 43(4):623--632, 1997.

\bibitem{BBL95}
O.~Berman, D.~J. Bertsimas, and R.~C. Larson.
\newblock Locating discretionary service facilities, II: Maximizing market
  size, minimizing inconvenience.
\newblock {\em OR}, 43(4):623--632, 1995.

\bibitem{BKX95}
O.~Berman, D.~Krass, and C.~W. Xu.
\newblock Locating flow-intercepting facilities: New approaches and results.
\newblock {\em Ann. Oper. Res.}, 60(1):121--143, 1995.

\bibitem{CDL05}
S.~Cabello, J.~M. D\'{\i}az-B{\'a}{\~n}ez, S.~Langerman, C.~Seara, and
  I.~Ventura.
\newblock Reverse facility location problems.
\newblock In {\em CCCG}, pages 68--71, 2005.

\bibitem{CDL10}
S.~Cabello, J.~M. D\'{\i}az-B{\'a}{\~n}ez, S.~Langerman, C.~Seara, and
  I.~Ventura.
\newblock Facility location problems in the plane based on reverse nearest
  neighbor queries.
\newblock {\em EJOR}, 202(1):99--106, 2010.

\bibitem{DZX05}
Y.~Du, D.~Zhang, and T.~Xia.
\newblock The optimal-location query.
\newblock In {\em SSTD}, pages 163--180, 2005.

\bibitem{FH09}
R.~Z. Farahani and M.~Hekmatfar.
\newblock {\em Facility location: concepts, models, algorithms and case
  studies}.
\newblock Contributions to Management Science. Physica-Verlag, 2009.

\bibitem{FGP93}
P.~Flajolet, G.~H. Gonnet, C.~Puech, and J.~M. Robson.
\newblock {Analytic variations on quadtrees}.
\newblock {\em Algorithmica}, 10(6):473--500, 1993.

\bibitem{F06}
D.~Fotakis.
\newblock Incremental algorithms for facility location and $k$-median.
\newblock {\em Theor. Comput. Sci.}, 361(2-3):275--313, 2006.

\bibitem{H81}
M.~J. Hodgson.
\newblock The location of public facilities intermediate to the journey to
  work.
\newblock {\em EJOR}, 6(2):199--204, 1981.

\bibitem{H90}
M.~J. Hodgson.
\newblock A flow capturing location-allocation model.
\newblock {\em Geographical Analysis}, 22(3):270--279, 1990.

\bibitem{JOD07b}
H.~Jia, F.~Ordóñez, and M.~M. Dessouky.
\newblock A modeling framework for facility location of medical services for
  large-scale emergencies.
\newblock {\em IIE Transactions}, 39:41--55, 2007.

\bibitem{JOD07a}
H.~Jia, F.~Ordóñez, and M.~M. Dessouky.
\newblock Solution approaches for facility location of medical supplies for
  large-scale emergencies.
\newblock {\em Computers and Industrial Engineering}, 52:257--276, 2007.

\bibitem{LAT07}
H.~Lahrmann, N.~Agerholm, N.~Tradi\v{s}auskas, and J.~Juhl.
\newblock Spar paa farten-an intelligent speed adaptation project in Denmark
  based on pay as you drive principles.
\newblock {\em European Congress on ITS and Services}, 2007.

\bibitem{M01}
A.~Meyerson.
\newblock Online facility location.
\newblock In {\em FOCS}, pages 426--431, 2001.

\bibitem{NP05}
S.~Nickel and J.~Puerto.
\newblock {\em Location theory: An unified approach}.
\newblock Springer, 2005.

\bibitem{pohl1971}
I.~Pohl.
\newblock Bi-directional search.
\newblock {\em Machine Intelligence}, 6:127--140, 1971.

\bibitem{TJL07}
N.~Tradi\v{s}auskas, J.~Juhl, H.~Lahrmann, and C.~S. Jensen.
\newblock Map matching for intelligent speed adaptation.
\newblock In {\em European Congress on ITS and Services}, 2007.

\bibitem{WOY09}
R.~C.-W. Wong, M.~T. {\"O}zsu, P.~S. Yu, A.~W.-C. Fu, and L.~Liu.
\newblock Efficient method for maximizing bichromatic reverse nearest neighbor.
\newblock {\em PVLDB}, 2(1):1126--1137, 2009.

\bibitem{XYL10}
X.~Xiao, B.~Yao, and F.~Li.
\newblock Optimal location queries in road network databases.
\newblock In {\em ICDE}, pages 804--815, 2011.

\bibitem{ZDX06}
D.~Zhang, Y.~Du, T.~Xia, and Y.~Tao.
\newblock Progressive computation of the min-dist optimal-location query.
\newblock In {\em VLDB}, pages 643--654, 2006.

\bibitem{ZWL10}
Z.~Zhou, W.~Wu, X.~Li, M.-L. Lee, and W.~Hsu.
\newblock MaxFirst for MaxBRkNN.
\newblock In {\em ICDE}, pages 828--839, 2011.

\end{thebibliography}

\end{document}